\documentclass[12pt,A4]{article}
\usepackage{amsmath}
\usepackage{amsthm}
\usepackage{amssymb}
\usepackage{color}
\usepackage{fancyhdr}
\usepackage{graphicx}

\chardef\coloryes=1 

\graphicspath{ {images/} }
\usepackage[margin=0.8in]{geometry}
\newtheorem{theorem}{Theorem}[section]

\newtheorem{lemma}[theorem]{Lemma}
\newtheorem*{remark}{Remark}
\theoremstyle{definition}
\newtheorem{definition}{Definition}[section]
\def\indeq{\qquad{}\!\!\!\!}  

\begin{document}

\title{Subsonic flow over a thin airfoil in ground effect}
\author{Mohamed Serry\\
Amjad Tuffaha}
\maketitle
\date{}

\bigskip

\indent Department of Mathematics and Statistics\\
\indent American University of Sharjah\\
\indent Sharjah, UAE\\
\indent e-mail:  mohamedserry91\char'100gmail.com \\
\indent e-mail: atufaha\char'100aus.edu

\begin{abstract}
  In this paper, the problem of compressible flow over a thin airfoil located near the ground is studied. A singular integral equation, also known as Possio equation \cite{Possio}, that relates the pressure jump along the airfoil to its downwash is derived. The derivation of the equation utilizes Laplace transform, Fourier transform,  method of images, and theory of Mikhlin multipliers.  The existence and uniqueness of solution to the Possio equation is verified for the steady state case and an approximate solution is obtained. The aerodynamic loads are then calculated based on the approximate solution. Moreover, the divergence speed of a continuum wing structure located near the ground is obtained based on the derived expressions for the aerodynamic loads.
\end{abstract}
\section{Introduction}
In the last few decades, the fields of aerodynamics and aero-elasticity have bloomed significantly due to advances in computation
power and experimentation.
 As a result, the implementation of analytical methods has receded slightly. Despite their limitations, analytical techniques have contributed significantly in the development of the fields of aerodynamics and aero-elasticity. An important example is the pioneering work of Theodorsen \cite{Theodorsen} who derived closed form expressions of the aerodynamic loads on thin airfoils in incompressible flow using tools from complex analysis. Despite the relative simplicity of the expressions derived by Theodorsen, they have been intensively used by a significant number of researchers to study the aero-elastic stability and control of wing structures (see for example  \cite{aeroelastic control1, aeroelastic control2, aeroelastic control3, aeroelastic control4}). Moreover, researchers have used Theodorsen's work as a basis to develop more accurate aerodynamic models (see for example \cite{Hodges book}).  It must to be noted that there were  also other early works, prior to Theodorsen work, that have  considered analytical expressions of aerodynamic loads in subsonic flow such as \cite{Lance, Sears, Watkins}. Another important example is  the relatively recent work of A.V. Balakrishnan that has revived the interest in analytical techniques in aerodynamics and aero-elasticity (see for example:  \cite{Balakrishnan book, Bal, iLiff paper}) and in particular the Possio equation of aeroelasticity (a generalization of the classical airfoil equation). Balakrishnan has implemented functional analytic techniques intensively to derive and solve singular integral equations from which the aerodynamic loads on thin airfoils in compressible flows can be obtained. Additionally, he studied the aero-elastic stability of wing structures using continuum models. Besides the work of Balakrishnan, there has been a series of recent mathematical works  (see for example: \cite{Carabineanu, Survey paper, Irena and Webster, Polyakov, Shubov, Shubov2, Webster}) which have studied the mathematical aspects (existence, uniqueness, obtaining solutions, and stability) of different aerodynamic and aeroelastic problems. In conclusion, implementation of analytical techniques, despite their limitation and complexity, is significant in studying different aerodynamic and aeroelastic phenomena.
 
	Ground effect is one of the aerodynamic phenomena that received attention from engineers and scientists for many years due to its significance in many applications \cite{ground effect literature}. The ground effect is an aerodynamic phenomenon that can be observed when a flying object is near the ground as the induced lift on the flying object becomes relatively high compared to the lift induced in an open flow. In this paper, approximations of analytical formulas for the aerodynamic lift and moment on a thin airfoil near the ground are obtained. The formulas are obtained by first deriving a singular integral equation, also known as the Possio equation \cite{avanzini}. The integral equation is derived based on linearized compressible potential flow theory, typical section theory, Fourier and Laplace transforms, and theory of Mikhlin multipliers \cite{Balakrishnan book, Balakrishnan Hindawi}. The solvability of the integral equation is then discussed for the steady state case and approximate solutions are then obtained for that case. Finally, the divergence speed of a continuum wing structure in steady compressible flow near the ground is calculated based on the derived approximate solutions.
 
\section{Problem Formulation}
In this work, subsonic flow over a thin airfoil that is located near the ground  is considered. The aim of the study is to derive formulas from which the aerodynamic loads on the airfoil can be calculated.

The airfoil is described by the set 
$\Gamma = [-b,b] \times \{ z=z_{0} \} \subset \mathbb{R}_{xz}^{+}$
at distance $z_{0}>0$ from the ground $z=0$ (see figure \eqref{fig:origional}), and is subject to an airflow with a free stream velocity of $U$ in the positive direction, so that $x=b$ signifies the trailing edge and $x=-b$ is the leading edge of the chord.
 \begin{figure}
\includegraphics[width=5in ,keepaspectratio=true]{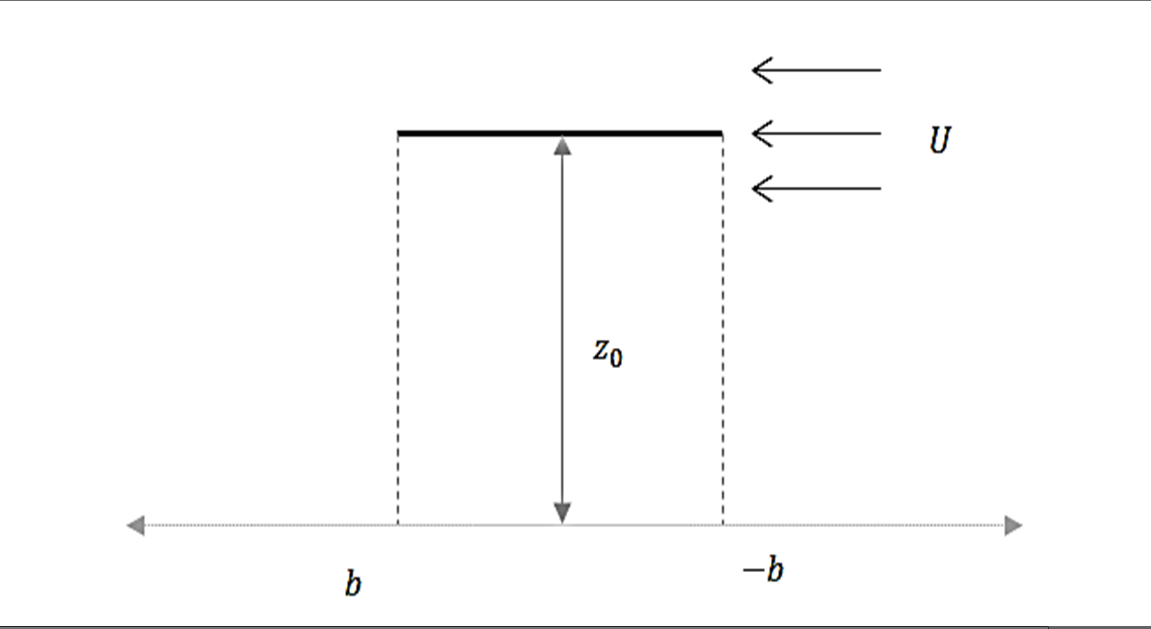}
\centering
\caption{Air flow over a thin  airfoil near the ground}
\label{fig:origional}
\end{figure}
\subsection{Flow Model}

The flow over the airfoil is assumed to be isentropic and following the behavior of ideal gases. Additionally, the movement of the airfoil is assumed to be small compared to the free stream velocity. The flow is assumed to be two dimensional based on typical section theory which is valid for wings with high aspect ratios. Accordingly, the flow over the airfoil is defined by a disturbance potential $\phi(x,z,t)$ that satisfies the differential equation
\begin{equation}\label{eq:flow equation}
\frac{\partial^2}{\partial t^2}\phi+2M a_{\infty}\frac{\partial^2}{\partial t \partial x}\phi=a_{\infty}^2(1-M^2)\frac{\partial^2}{\partial x^2}\phi+a_{\infty}^2\frac{\partial^2}{\partial z^2}\phi,
\end{equation}
where $-\infty<x<\infty$, $0\leq z<\infty$, $t\geq0$, $a_{\infty}$ is the free stream speed of sound and $M=U/a_{\infty}$ is the free stream Mach number. Note that the assumption of potential flow neglects the effects of the boundary layer which become significant when the airfoil is very close to the ground. Therefore, the implementation of the theory discussed in this work is assumed to be valid when the airfoil is at an elevation range in which the boundary layer effects are small but the ground effect is still present. Additionally, equation \eqref{eq:flow equation} is valid for ranges of Mach number between 0 and 0.7 \cite{iLiff report}. To calculate the pressure jump along the airfoil, a reasonable approximation in terms of the acceleration potential $\psi(x,z,t)$ is used where
\begin{equation}\label{eq:acceleration potential}
  \psi=\frac{\partial}{\partial t}\phi+U\frac{\partial}{\partial x}\phi.
\end{equation}The pressure jump term $A(x,t)$ is then defined by
\begin{equation}\label{eq:pressure jump}
A=-\frac{\Delta \psi}{U},
\end{equation}where
\begin{equation}\nonumber
\Delta \psi=\psi(z_{0}^+)-\psi(z_{0}^-).
\end{equation}
Moreover, equation \eqref{eq:flow equation} is supplemented with the following boundary conditions which describe zero normal velocity at the ground,  matching normal flow-structure velocity on the airfoil (flow tangency), zero pressure jump off the wing (Kutta -Joukowski condition), zero pressure jump at the trailing edge of the airfoil (Kutta condition), and vanishing disturbance potential far from the airfoil.
\begin{flalign}\label{eq:boundary conditions}
\text{Zero normal velocity at the ground:}~~&\frac{\partial}{\partial z}\phi=0,~z=0, \\\nonumber
\text{Flow tangency condition:}~~&\frac{\partial}{\partial z}\phi =w_{a},~~ z=z_{0}~~\text{and}~~|x|\leq b,\\\nonumber
\text{Kutta-Joukowski condition:}~~&A(x,t)=0,~|x|>b\\\nonumber
\text{Kutta condition:}~~&\lim_{x\rightarrow{b^-}}A(x,t)=0,\\\nonumber
\text{Vanishing  disturbance potential at infinity :}~~&\lim_{x\rightarrow \pm\infty,z\rightarrow \infty}\phi(x,z,t)=0,
\end{flalign}
where $w_{a}(x,t)$ is the downwash or the normal velocity on the airfoil surface. Note that the Kutta condition insures the uniqueness of solution to the problem under consideration as shall be discussed in the later sections .
\section{Derivation of the Possio Equation}
In this section, an equation  that relates the pressure jump along the airfoil chord to the airfoil downwash, which is referred to as a Possio equation, is derived. The derivation process starts with applying the Laplace transform in the t variable and the Fourier transform in the $x$ variable  to both sides of equation (\ref{eq:flow equation}) to obtain
\begin{equation}\label{eq:transformed flow equation}
\lambda^2 \hat{\hat{\phi}}+2 M a_{\infty}i \omega \lambda \hat{\hat{\phi}}=-a_{\infty}^2(1-M^2)\omega^2\hat{\hat{\phi}}+a_{\infty}^2\frac{\partial^2}{\partial z^2}\hat{\hat{\phi}},
\end{equation}
where $\hat{f}(x,z,\lambda)=\int_{0}^{\infty}e^{-\lambda t}f(x,z,t) \, dt$, $Re(\lambda)\geq \sigma>0$ and $\hat{\hat{f}}(\omega,z,\lambda)=\int_{-\infty}^{\infty}e^{-i \omega x}\hat{f}(x,z,\lambda)\, dx$. Rearranging equation (\ref{eq:transformed flow equation}) results in
\begin{equation}\label{eq:transformed flow equation simplified}
  \frac{\partial^2}{\partial z^2}\hat{\hat{\phi}}=B(\omega,k)\hat{\hat{\phi}},
\end{equation}
where $B(\omega,k)=M^2(k+i\omega)^2+\omega^2$ and $k=\frac{\lambda}{U}$ is the reduced frequency. Since, $Re(k)\geq \sigma >0 $,  it can be shown that the function $B(k,\omega)$ is never $0$. 

Next, the method of images is used to account for the ground effect (see figure~\eqref{fig:image method}). The method assumes an open flow (no ground) and an image of the airfoil to be located at a distance  $-z_{0}$ from the ground axis.
\begin{figure}
\includegraphics[width=5in ,keepaspectratio=true]{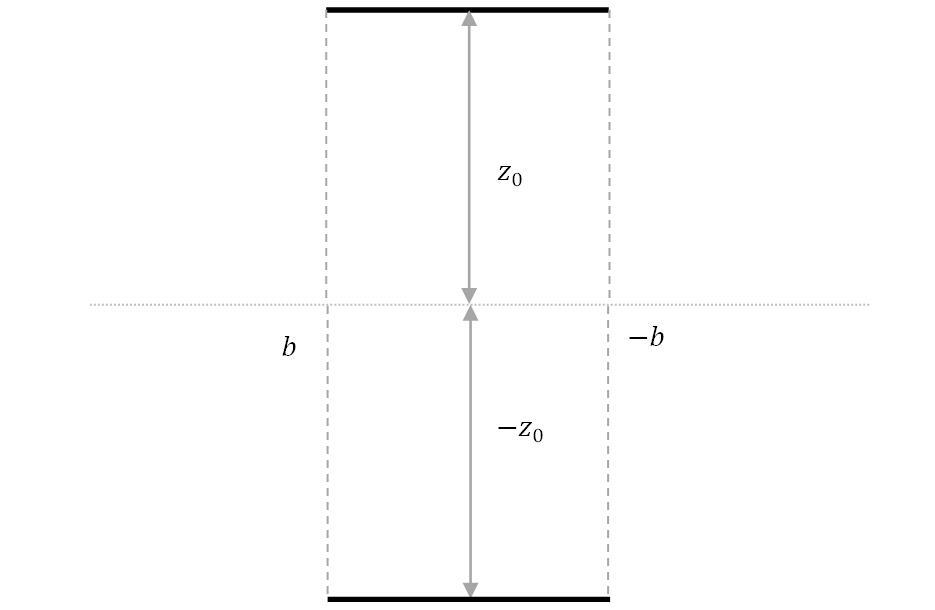}
\centering
\caption{The flow over the airfoil and its image in an open flow}
\label{fig:image method}
\end{figure}
Due to the linearity of equation \eqref{eq:transformed flow equation simplified}, the solution of the flow problem is obtained by studying the open flow over the airfoil and its image separately.
The separate solutions are given by
\begin{eqnarray}\nonumber
\phi_{airfoil}=\begin{cases}
           \hat{\hat{\phi}}(z_{0}^+)e^{-\sqrt{B(\omega,k)}(z-z_{0})}, & z>z_{0} \\
           \hat{\hat{\phi}}(z_{0}^-)e^{\sqrt{B(\omega,k)}(z-z_{0})}, & z<z_{0}
         \end{cases} & \phi_{image}=\begin{cases}\hat{\hat{\phi}}(z_{0}^-)e^{-\sqrt{B(\omega,k)}(z+z_{0})}, & z>-z_{0} \\  \hat{\hat{\phi}}(z_{0}^+)e^{\sqrt{B(\omega,k)}(z+z_{0})}, & z<-z_{0}\end{cases}
\end{eqnarray}
The normal velocities about the axis $z=z_{0}$ are given by
\begin{eqnarray}
\label{v+}
\hat{\hat{v}}_+=&-\sqrt{B(\omega,k)}\left(\hat{\hat{\phi}}(z_{0}^+)+\hat{\hat{\phi}}(z_{0}^-)e^{-2\sqrt{B(\omega,k)}z_{0}}\right),\\
\label{v-}
\hat{\hat{v}}_-=&\sqrt{B(\omega,k)}\hat{\hat{\phi}}(z_{0}^-)\left(1-e^{-2\sqrt{B(\omega,k)}z_{0}}\right).
\end{eqnarray}
The solution to the flow problem should insure the continuity of the velocity field. Therefore, $\hat{\hat{v}}_+=\hat{\hat{v}}_-=\hat{\hat{v}}$. Using the equations in \eqref{v+} and \eqref{v-}, the difference in $\hat{\hat{\phi}}(z_{0}^+)$ and $\hat{\hat{\phi}}(z_{0}^-)$ is expressed in terms of $\hat{\hat{v}}$ as

\begin{equation}\label{eq:transformed difference in the potential}
 \hat{\hat{\phi}}(z_{0}^+)- \hat{\hat{\phi}}(z_{0}^-)=\frac{-2}{\sqrt{B(\omega,k)}\left(1-e^{-2\sqrt{B(\omega,k)}z_{0}}\right)}\hat{\hat{v}}.
\end{equation}
Next, the Fourier and the Laplace Transforms are applied to the acceleration potential in (\ref{eq:acceleration potential}) and the pressure jump term in (\ref{eq:pressure jump}) to obtain
\begin{eqnarray}\label{eq:transformed acceleration potential and pressure jump}
   \hat{\hat{\psi}}=&(\lambda+i U\omega) \hat{\hat{\phi}},\\
   \hat{\hat{A}}=&- \frac{\Delta \hat{\hat{\psi}}}{U}.
\end{eqnarray}\label{}
Using equations \eqref{eq:transformed difference in the potential} and \eqref{eq:transformed acceleration potential and pressure jump}, the pressure jump term $\hat{\hat{A}}$ is represented in terms of the normal velocity $\hat{\hat{v}}$ as 
  \begin{equation}\nonumber
  \hat{\hat{A}}=\frac{(2k+i\omega)}{\sqrt{B(\omega,k)}\left(1-e^{-2\sqrt{B(\omega,k)}z_{0}}\right)}\hat{\hat{v}}.
\end{equation}
Rearranging the above equation yields:
\begin{equation}\label{eq:possio equation in Fourier domain}
\hat{\hat{v}}=\frac{\sqrt{B(\omega,k)}\left(1-e^{-2\sqrt{B(\omega,k)}z_{0}}\right)}{2(k+i\omega)}\hat{\hat{A}}.
\end{equation}
Equation~\eqref{eq:possio equation in Fourier domain} is the desired Possio equation that relates the pressure jump to the normal velocity at $z_0$ in the Fourier domain. The next step is to obtain an integral equation, based on equation \eqref{eq:possio equation in Fourier domain}, that relates the pressure jump along the airfoil to the airfoil downwash. To accomplish this, we appeal to the theory of Mikhlin multipliers in the next section.

\section{Mikhlin Multipliers}
In this section, the  necessary  definition and theory of Mikhlin multipliers used in this work are presented. The readers are refereed to \cite{Mikhlin} for detailed discussions and proofs.
\begin{definition}\label{def:mikhlin multiplier}
  Let $f$ and $g$ be two functions in $L^{p}(-\infty,\infty)$ where $p>1$ and let their Fourier transforms $F$ and $G$ be related by
  \begin{equation}\nonumber
    G(\omega)=\mu(\omega) F(\omega),
  \end{equation}
  where $\mu$ is $C^{1}$ and satisfies
  \begin{equation}\nonumber
    |\mu(\omega)|+|\omega \mu^{'}(\omega)|<C<\infty
  \end{equation}
  for all $\omega$ except maybe at $\omega=0$. Then $\mu$ is called a Mikhlin multiplier
\end{definition}
\begin{theorem}
  If $f$ and $g$ are two functions defined and related as in definition \ref{def:mikhlin multiplier}, then there exists a bounded linear operator $T:L^{p}(-\infty,\infty)\rightarrow L^{p}(-\infty,\infty) $ where $p>1$ such that
  \begin{equation}\nonumber
    g=T(f).
  \end{equation}
\end{theorem}
\begin{proof}
The detailed proof is discussed in the work of Mikhlin \cite{Mikhlin}.
\end{proof}
The next step is to show that the multiplier $\frac{\sqrt{B(\omega,k)}(1-e^{-2\sqrt{B(\omega,k)}z_{0}})}{2(k+i\omega)}$  in equation (\ref{eq:possio equation in Fourier domain}), which we denote by $\gamma(\omega,k)$, is a Mikhlin multiplier.
\begin{theorem}
  $\gamma (\omega)$ is a Mikhlin multiplier
\end{theorem}
\begin{proof}
$\gamma(\omega)$ can be written as $\gamma(\omega)=\alpha(\omega)\beta(\omega)$ where $\alpha(\omega)=\frac{\sqrt{B(\omega,k)}}{2(k+i\omega)}$ and $\beta(\omega)=1-e^{-2\sqrt{B(\omega,k)}z_0}$. It was shown in \cite{Balakrishnan Hindawi} that $\alpha(\omega)$ is a Mikhlin multiplier. Therefore, it remains to show that $\beta(\omega)$ is a Mikhlin multiplier. 

In particular, the function $|\beta(\omega)|$ satisfies the estimate
\begin{equation}\label{}
  |\beta(\omega)|=|1-e^{-2\sqrt{B(\omega,k)}z_0}|\leq 2.
\end{equation}
The term $|\omega\beta^{'}(\omega)|$ can be written as
\begin{equation}
|\omega\beta^{'}(\omega)|=\frac{|\omega B^{'}(\omega,k)|}{|z_{0}\sqrt{B(\omega,k)}|}|e^{-2\sqrt{B(\omega,k)}z_0}|.
\end{equation}
For a fixed value of $k$, it can be verified, by calculations that are omitted in this paper, that
\begin{itemize}
  \item $|\omega B^{'}(\omega,k)|$ is asymptotically equivalent to $2(1-M^2)\omega^2$.
  \item $|\sqrt{B(\omega,k)}|$ is asymptotically equivalent to $\sqrt{1-M^2}|\omega|$.
  \item $e^{-2\sqrt{B(\omega,k)}z_0}|$ is asymptotically equivalent to $e^{-2z_{0}\sqrt{1-M^2}|\omega|}$.
  \item Therefore, $|\omega\beta^{'}(\omega)|$ is continuous and has convergent limits at infinity, and hence $|\omega\beta^{'}(\omega)|$ is bounded.
\end{itemize}
Therefore, $\beta(\omega)$ is a Mikhlin multiplier.
\end{proof}
Based on the previous discussion, there exists a bounded linear operator $T:L^{p}(-\infty,\infty)\rightarrow L^{p}(-\infty,\infty)$ corresponding to equation \eqref{eq:possio equation in Fourier domain} such that
\begin{equation}\label{eq:general integral equation}
\hat{v}=T(\hat{A}).
\end{equation}
Applying the projection operator $\mathcal{P}:L^{p}(-\infty,\infty)\rightarrow L^{p}[-b,b]$ to both sides of equation \eqref{eq:general integral equation} results in
\begin{equation}\label{eq:general possio equation}
{\hat{w}}_a=\mathcal{P}T(\hat{A}),
\end{equation}
which is the Possio equation that relates the pressure jump to the downwash. The existence and uniqueness of solution to equation \eqref{eq:general possio equation} depends on the properties of the operator $T$. In the next section, the existence and uniqueness of solution to equation \eqref{eq:general possio equation} is discussed for the steady state case.

\section{Special Case k=0}
In this section, the solvability of equation \eqref{eq:possio equation in Fourier domain} is considered in the case of $k=0$, which corresponds to the steady state case. Setting $k=0$ in equation \eqref{eq:possio equation in Fourier domain} results in
\begin{equation}\label{eq:possioequation in Fourier domain k=0}
\hat{\hat{v}}=\frac{\sqrt{1-M^2}|\omega|\left(1-e^{-2\sqrt{1-M^2}|\omega|z_{0}}\right)}{2i\omega}\hat{\hat{A}}.
\end{equation}
The multiplier $|\omega|/i\omega$ corresponds to the Hilbert operator $\mathcal{H}$ which is defined as
\begin{equation}\label{eq:helbert operator}
\mathcal{H}(f)(t)=\frac{1}{\pi}\int_{-\infty}^{\infty}\frac{f(\tau)}{t-\tau}\, d\tau.
\end{equation}
Additionally, the multiplier $e^{-2\sqrt{1-M^2}|\omega|z_{0}}$ corresponds to an integral operator $\mathcal{L}$ that is defined as
\begin{equation}\label{eq:integral operator L}
  \mathcal{L}(f)(t)=\frac{1}{\pi c}\int_{-\infty}^{\infty}\frac{f(\tau)}{1+(\frac{t-\tau}{c})^2} \, d\tau,
\end{equation}
where
\begin{equation}
c=2z_{0}\sqrt{1-M^2}.
\end{equation}
 Note that the parameter $c$ plays an essential role in the existence of solution argument. Based on equation \eqref{eq:possioequation in Fourier domain k=0}, the operators defined in \eqref{eq:helbert operator} and \eqref{eq:integral operator L}, and applying the projection operator, we obtain following integral equation:
\begin{equation}\label{eq:possio integral equation k=0}
  \frac{2}{\sqrt{1-M^2}}w_{a}=\mathcal{P}\mathcal{H}(\mathcal{I}-\mathcal{L})A.
\end{equation}
Note that $A$ vanishes off the chord. Therefore, $A=\mathcal{P}A$. Additionally, the operators $\mathcal{H}$ and $(\mathcal{I}-\mathcal{L})$ commute as their product in the Fourier domain correspond to a Mikhlin multiplier. Implementing the above points and the distributive property of operators in equation \eqref{eq:possio integral equation k=0} results in
\begin{equation}\label{eq:possio integral equation k=0 expanded}
  \frac{2}{\sqrt{1-M^2}}w_{a}=\mathcal{P}\mathcal{H}\mathcal{P}A-\mathcal{P}\mathcal{L}\mathcal{H}\mathcal{P}A.
\end{equation}
The operator $\mathcal{L}\mathcal{H}\mathcal{P}$ is given by
\begin{equation}\label{eq:LHP}
\mathcal{L}\mathcal{H}\mathcal{P}(f)(x)=\frac{1}{\pi c}\int_{-\infty}^{\infty}\frac{1}{1+\left(\frac{x-t}{c}\right)^2}\left(\frac{1}{\pi}\int_{-b}^{b}\frac{f(\tau)}{t-\tau}d\tau\right)\, dt.
\end{equation}
Changing the order of integration in \eqref{eq:LHP} results in
\begin{equation}
\mathcal{L}\mathcal{H}\mathcal{P}(f)(x)=\frac{-1}{\pi c}\int_{-b}^{b}f(\tau)\left(\frac{1}{\pi}\int_{-\infty}^{\infty}\frac{1}{\tau-t}\frac{1}{1+\left(\frac{t-x}{c}\right)^2}dt\right )\, d\tau.
\end{equation}
The term $\frac{1}{\pi}\int_{-\infty}^{\infty}\frac{1}{\tau-t}\frac{1}{1+\left(\frac{t-x}{c}\right)^2}\,dt$ corresponds to the Hilbert operator applied to  $ h(t)=\frac{1}{1+\left(\frac{t-x}{c}\right)^2}$. We then note that $\mathcal{H}(\frac{1}{1+t^2})=\frac{t}{1+t^2}$ and the positive shifting and scaling are preserved under the Hilbert transform. Finally, $\mathcal{P}\mathcal{L}\mathcal{H}\mathcal{P}$ can be written as
\begin{equation}\label{eq:PLHP}
\mathcal{P}\mathcal{L}\mathcal{H}\mathcal{P}(f)(x)=\frac{-1}{\pi c}\int_{-b}^{b}f(\tau)g\left(\frac{\tau-x}{c}\right)\, d\tau,~~~|x|\leq b,
\end{equation}
where $g(t)=\frac{t}{1+t^2}$. The operator $\mathcal{P}\mathcal{H}\mathcal{P}$ corresponds to the finite Hilbert operator $\mathcal{H}_b$ which is given by
\begin{equation}\label{eq:finite hilbert operator}
\mathcal{H}_b(f)(t)=\frac{1}{\pi}\int_{-b}^{b}\frac{f(\tau)}{t-\tau}\, d\tau,~~|x|\leq b.
\end{equation}
The inverse of the finite Hilbert operator exists if the Kutta condition described by the third equation in (\ref{eq:boundary conditions}) is imposed. 

Therefore, the Possio equation has the form
\begin{equation}\label{Final Possio}
  \frac{2}{\sqrt{1-M^2}}w_{a}=(\mathcal{H}_{b}- \mathcal{P}\mathcal{L}\mathcal{H}\mathcal{P})A.
\end{equation}

\subsection{ Inversion of the  Finite Hilbert Transform}

In this work,  the Tricomi operator $\mathcal{T}$ defined by
\begin{equation}\label{eq:tricomi operator}
\mathcal{T}(f)(x):=\frac{1}{\pi}\sqrt{\frac{b-x}{b+x}}\int_{b}^{b}\sqrt{\frac{b+\tau}{b-\tau}}\frac{f(\tau)}{x-\tau}\, d\tau,~~|x|\leq b,
\end{equation}
is used as an inversion formula of the finite Hilbert transform.

\begin{lemma}\label{inversion}
Given $f \in L^{p}[-b,b]$ with $p>4/3$ there exists a solution $g  \in L^{r}[-b,b]$ for all $r < 4/3$ to the equation $$\mathcal{H}_{b}(g)=f.$$
Moreover, for $p>1$, any solution $g$  has the form
\begin{equation}
g(x)= \frac{1}{\pi}\int_{-b}^{b}\sqrt{\frac{b^{2}-y^{2}}{b^{2}-x^{2}}}\frac{f(y)}{x-y}\, dy + \frac{C}{\sqrt{b^{2}-x^{2}}},
\end{equation}
with $C$ being an arbitrary constant.
\begin{proof}
This is a classical result due to Tricomi and Sohngen, \cite{Sohngen, Tricomi}.
\end{proof}
\end{lemma}

Next, a bound for the norm of the Tricomi operator is obtained in the following lemma.
\begin{lemma}\label{Tricomi bounded}
The Tricomi operator $\mathcal{T}$ defined in \eqref{eq:tricomi operator} is bounded from $L^{p}[-b,b]$ to $L^{r}[-b,b]$ for every $p>2$ and $1\leq r <4/3$.
\end{lemma}\label{lemma:T bounded}
\begin{proof}
  $\mathcal{T}(f)$ can be written as the following:
\begin{equation*}
\begin{split}
\mathcal{T}(f)(x)&=\frac{1}{\pi}\sqrt{\frac{b-x}{b+x}}\int_{-b}^{b}\sqrt{\frac{b+y}{b-y}}\frac{f(y)}{y-x}\, dy \\
&=\frac{1}{\pi}\int_{-b}^{b}\sqrt{\frac{b^2-y^2}{b^2-x^2}}\frac{b-x}{b-y}\frac{f(y)}{y-x} \, dy\\
&=\frac{1}{\pi}\int_{-b}^{b}\sqrt{\frac{b^2-y^2}{b^2-x^2}}\frac{f(y)}{y-x}dy+\frac{1}{\pi}\int_{-b}^{b}\sqrt{\frac{b^2-y^2}{b^2-x^2}}\frac{f(y)}{b-y}\, dy\\
&=-\frac{1}{\pi}\frac{1}{\sqrt{b^2-x^2}}\int_{-b}^{b}\frac{(x+y)f(y)}{\sqrt{b^2-x^2}+\sqrt{b^2-y^2}}\, dy\\
&  \indeq  -\frac{1}{\pi}\int_{-b}^{b}\frac{f(y)}{x-y}dy+\frac{1}{\pi}\frac{1}{\sqrt{b^2-x^2}}\int_{-b}^{b}\sqrt{\frac{b+y}{b-y}}f(y) \, dy\\
&=-\frac{1}{\sqrt{b^2-x^2}}\Pi_{1}(x)-\Pi_{2}(x)+\frac{1}{\sqrt{b^2-x^2}}\Pi_{0}
\end{split}
\end{equation*}
where $$\Pi_{1}(x)=\frac{1}{\pi}\int_{-b}^{b}\frac{(x+y)f(y)}{\sqrt{b^2-x^2}+\sqrt{b^2-y^2}}\, dy,$$ 

$$\Pi_{2}(x)=\frac{1}{\pi}\int_{-b}^{b}\frac{f(y)}{x-y}\,dy,$$ 
and 
$$\Pi_{0}=\frac{1}{\pi}\int_{-b}^{b}\sqrt{\frac{b+y}{b-y}}f(y)\, dy.$$ 
Next, we estimate the norm of $\mathcal{T} f$ in $L^{r}[-b,b]$ 
given $f \in L^{p}[-b,b]$.
 Using Minkowski's inequality, we have
\begin{equation}\label{eq:minkowski inequality}
||\mathcal{T}(f)||_{L^{r}[-b,b]}\leq \left|\left|\frac{1}{\sqrt{b^2-x^2}}\Pi_{1}\right|\right|_{L^{r}[-b,b]}
+\left|\left|\Pi_{2}\right|\right|_{L^{r}[-b,b]}
+|\Pi_{0}|\left|\left|\frac{1}{\sqrt{b^2-x^2}}\right|\right|_{L^{r}[-b,b]}.
\end{equation}

In the following, the definitions of $r$ and $p$ are unchanged but the terms $p'$ and $q'$ are redefined for each subsection.  Next, we consider  each term in \eqref{eq:minkowski inequality} separately

\subsubsection{First term $\left|\left|\frac{1}{\sqrt{b^2-x^2}}\Pi_{1}\right|\right|_{L^{r}[-b,b]}$ }
The first term $\left|\left|\frac{1}{\sqrt{b^2-x^2}}\Pi_{1}\right|\right|_{L^{r}[-b,b]}$ can be estimated by H\"older's inequality so that
\begin{equation}\label{}
  \left|\left|\frac{1}{\sqrt{b^2-x^2}}\Pi_{1}\right|\right|_{L^{r}[-b,b]}^r\leq\left|\left|\frac{1}{\sqrt{b^2-x^2}}\right|\right|_{L^{q'}[-b,b]}^{\frac{q'}{r}}
  		\left|\left|\Pi_{1}\right|\right|_{L^{p'}[-b,b]}^{\frac{p'}{r}}
\end{equation}
where $\frac{1}{r}=\frac{1}{q'}+\frac{1}{p'}$. Note that $\left|\left|\frac{1}{\sqrt{b^2-x^2}}\right|\right|_{L^{q'}[-b,b]}$ is finite if $q'< 2$. Moreover, the term  $\left| \left|\Pi_{1}\right|\right|_{L^{p'}[-b,b]}$ can be estimated using H\"older's inequality 
 \begin{equation}
 \begin{split}
 \int_{-b}^{b} |\Pi_{1}(x)|^{p'} \, dx & \leq 
 			 \left( \int_{-b}^{b}\int_{-b}^{b}  \left(\frac{|x+y|}{\pi (\sqrt{b^2-x^2} +\sqrt{b^2-y^2})}\right)^{p'} \, dy \, dx \right) 
				\left(\int_{-b}^{b}|f(y)|^{p} \, dy\right)^{\frac{p'}{p}}\\
 &\leq \left(\int_{-b}^{b} \int_{-b}^{b} \left(\frac{|x+y|}{\pi(\sqrt{b^2-x^2}+\sqrt{b^2-y^2})} \right)^{p'} \,  dy \,  dx \right) \Vert f \Vert_{L^{p}[-b,b]}^{p'} ,
 \end{split}
 \end{equation}
where $\frac{1}{p'}+\frac{1}{p}=1$. The double integral can be shown to be finite for $p'< 4$, \cite{Tricomi}.

To illustrate that, the double integral is estimated as the following:
\begin{equation}\label{eq:1st bound of tricomi }
\int_{-b}^{b}\int_{-b}^{b}\left(\frac{|x+y|}{\sqrt{b^2-x^2}+\sqrt{b^2-y^2}} \right )^{p'}dx\, dy \leq \int_{-b}^{b}\int_{-b}^{b}\left(\frac{2b}{\sqrt{b^2-x^2-y^2}} \right )^{p'} \, dx\,dy.
\end{equation}
Using polar coordinates, the term $\int_{-b}^{b}\int_{-b}^{b}\left(\frac{2b}{\sqrt{b^2-x^2-y^2}} \right )^{p'}\,dx\, dy$ can be written as
\begin{equation}\label{}
 \int_{-b}^{b}\int_{-b}^{b}\left(\frac{2b}{\sqrt{b^2-x^2-y^2}} \right )^{p'}\, dx \, dy= \frac{8(2b)^{p'}}{p'-2}\left(b^{2-p'}\int_{0}^{\frac{\pi}{4}}\left(2-\sec^2\theta \right )^{1-\frac{p'}{2}}\, d\theta-\frac{\pi}{2^{1+\frac{p'}{2}}}b^{2+p'} \right).
\end{equation}
Using the substitution $u=1-\tan^2\theta$, the integration term $\int_{0}^{\frac{\pi}{4}}\left(2-\sec^2\theta \right )^{1-\frac{p'}{2}}\, d\theta$ can be estimated as
\begin{equation}\label{eq:estimate for the integral 1-sec}
\begin{split}
\int_{0}^{\frac{\pi}{4}}\left(2-\sec^2\theta \right )^{1-\frac{p'}{2}}d\theta &=\int_{0}^{\frac{\pi}{4}}\frac{1}{(1-\tan^2\theta)^{\frac{p'}{2}-1}}\,d\theta\\
&=\frac{1}{2}\int_{0}^{1}\frac{1}{u^{\frac{p'}{2}-1}\sqrt{1-u}(2-u)} \, du\\
&=\frac{1}{2}\int_{0}^{\frac{1}{2}}\frac{1}{u^{\frac{p'}{2}-1}\sqrt{1-u}(2-u)} \, du+\frac{1}{2}\int_{\frac{1}{2}}^{1}\frac{1}{u^{\frac{p'}{2}-1}\sqrt{1-u}(2-u)}\, du\\
&\leq\frac{\sqrt{3}}{2}\int_{0}^{\frac{1}{2}}\frac{1}{u^{\frac{p'}{2}-1}}du+2^{\frac{p'}{2}-2}\int_{\frac{1}{2}}^{1}\frac{1}{\sqrt{1-u}}\, du.\\
\end{split}
\end{equation}
From the inequality, we have finiteness when $\frac{p'}{2}-1<1$ and hence $p'<4$. 
Since  $q'<2$, we have $r<\frac{4}{3}$ and $p>\frac{4}{3}$.

\subsubsection{The second term $\left|\left|\Pi_{2}\right|\right|_{L^{r}[-b,b]}$}
The term $\Pi_{2}$ corresponds to the finite Hilbert operator applied on $f$. Using Riesz theorem \cite{Riesz}, the finite Hilbert operator is bounded on $L^s[-b,b]$ for any  $s>1$. Consequently, we have $\left|\left|\Pi_{2}\right|\right|_{L^{r}[-b,b]} \leq C_{p,r} \Vert f \Vert_{L^{p}[-b,b]} $ for  $ 1 < r \leq p$.

\subsubsection{The third term  $\Pi_{0}\left|\left|\frac{1}{\sqrt{b^2-x^2}}\right|\right|_{L^{r}[-b,b]}$}
 Finally, for the last term $\Pi_{0}\left|\left|\frac{1}{\sqrt{b^2-x^2}}\right|\right|_{L^{r}[-b,b]}$, we have that $\Pi_{0}$ is a finite constant since   
\begin{equation}\label{}
\begin{split}
|\Pi_{0}|&\leq \frac{1}{\pi}\int_{-b}^{b}\left|\sqrt{\frac{b+y}{b-y}}f(y)\right|dy\\
&\leq\frac{1}{\pi} \left(\int_{-b}^{b}\left(\frac{b+y}{b-y}\right)^\frac{p'}{2}dy\right)^{\frac{1}{p'}} \Vert f \Vert_{L^{p}[-b,b]}^{p},\\
\end{split}
\end{equation}
where $\frac{1}{p'}+\frac{1}{p}=1$. Note that the first term on the right side of the inequality is finite when $p'<2$ which corresponds to the case of $p>2$. Moreover, $\left|\left|\frac{1}{\sqrt{b^2-x^2}}\right|\right|_{L^{r}[-b,b]}$ is finite when $r<2$ without imposing any conditions on $p$. Based on studying the boundedness of the three terms in \eqref{eq:minkowski inequality}, $\mathcal{T}$ is bounded when $p>2$ and $r<\frac{4}{3}$ and that completes the proof.
\end{proof}
We next characterize the solution of the equation $\mathcal{H}_{b}(g)= f$  using the Tricomi operator $\mathcal{T}$ defined in \eqref{eq:tricomi operator}.

\begin{lemma}\label{T}
Given $f \in L^{p}[-b,b]$ with $p >2$, there exists a  solution to the equation 
$$\mathcal{H}_{b}(g)=f$$
 given by $g =\mathcal{T}(f) \in L^{4/3-}[-b,b]$ . Moreover, the solution is unique in the class of functions satisfying the Kutta condition  $\lim_{x \to b^{-}} g(x)=0$.
\begin{proof}
Starting with the inversion formula from Lemma \ref{inversion}, we can express any solution $g(x)$ as
\begin{align*}
g(x)& =\frac{1}{\pi}\int_{-b}^{b}\sqrt{\frac{b^2-y^2}{b^2-x^2}}\frac{f(y)}{y-x} \,dy
	+ \frac{C}{\sqrt{b^{2}-x^{2}}} \\
 & =  \mathcal{T}f - \frac{1}{\pi}\int_{-b}^{b}\sqrt{\frac{b^2-y^2}{b^2-x^2}}\frac{f(y)}{b-y}\, dy + \frac{C}{\sqrt{b^{2}-x^{2}}}\\
& = \mathcal{T}f  
-\frac{1}{\pi}\frac{1}{\sqrt{b^2-x^2}}\int_{-b}^{b}\sqrt{\frac{b+y}{b-y}}f(y) \, dy + \frac{C}{\sqrt{b^{2}-x^{2}}}
\\
& = \mathcal{T}f  
-\frac{C_{1}}{\sqrt{b^2-x^2}} + \frac{C}{\sqrt{b^{2}-x^{2}}},
\end{align*}
where $C_{1} = \frac{1}{\pi}\int_{-b}^{b}\sqrt{\frac{b+y}{b-y}}f(y) \, dy$
is finite since $f \in L^{p}[-b,b]$ for $p>2$, which implies all possible solutions $g(x)$ can be expressed as
\begin{align*}
g(x) = \mathcal{T}(f)(x)  
+\frac{C_{0}}{\sqrt{b^2-x^2}},
\end{align*}
with $C_{0}$ arbitrary. Now, since we seek a solution $g(x)$ satisfying the Kutta condition, this implies $C_{0}=0$, and the proof is complete.
\end{proof}
\end{lemma}
{\bf Remark}: A solution to the equation and its uniqueness in a class of functions satisfying the Kutta condition can also be obtained using the fact that $H_{b}$ is a bijective isometry from the weighted Hilbert space $L^{2}(\sigma)$ on itself where $\sigma(x) = \frac{\sqrt{b+x}}{\sqrt{b-x}}$, \cite{Schleiff, Okada, OkadaPr}.

\subsection{Existence and Uniqueness of Solution of the Possio equation for the case k=0}

In this section, the existence and uniqueness of solution to equation \eqref{Final Possio} is established through the framework of contraction mappings theory. 
 Applying the Tricomi operator to each side of equation \eqref{Final Possio} yields
\begin{equation}\label{eq:final form of Possio equation to be solved}
 \frac{2}{\sqrt{1-M^2}}\mathcal{T}(w_{a})=A-\mathcal{T}\mathcal{P}\mathcal{L}\mathcal{H}\mathcal{P}A=(\mathcal{I}-\mathcal{T}\mathcal{P}\mathcal{L}\mathcal{H}\mathcal{P})A.
\end{equation}
The equivalence of equation \eqref{Final Possio} and \eqref{eq:final form of Possio equation to be solved}
follows from Lemma \ref{T}, under the assumption of $w_{a} \in L^{2+}[-b,b]$ and $A \in L^{p}[-b,b]$  with $A$ satisfying the Kutta condition.  
Note that the Tricomi operator defined in \eqref{eq:tricomi operator} maps functions from $L^{2+}[-b,b]$ to $L^{\frac{4}{3}-}[-b,b]$ while the operator $\mathcal{P}\mathcal{L}\mathcal{H}\mathcal{P}$ maps functions from $L^{\frac{4}{3}- }[-b,b]$ to $C[-b,b]$ boundedly. Hence, we establish existence of a solution $A$ in $L^{\frac{4}{3}- }[-b,b]$ to \eqref{eq:final form of Possio equation to be solved}. If there exists values of $c$ such that $||\mathcal{T}\mathcal{P}\mathcal{L}\mathcal{H}\mathcal{P}||_{\mathcal{L}(L^{\frac{4}{3}-}[-b,b])}<1$, then the solution to equation \eqref{eq:final form of Possio equation to be solved} is given uniquely by
\begin{equation}\label{eq:solution to Possio equation k=0}
A=\frac{2}{\sqrt{1-M^2}}\sum_{n=0}^{\infty}(\mathcal{T}\mathcal{P}\mathcal{L}\mathcal{H}\mathcal{P})^n\mathcal{T}(w_a).
\end{equation}
The expression in \eqref{eq:solution to Possio equation k=0} can then be rearranged to give
\begin{equation}\label{eq:solution to possio equation k=0 seperated}
A=\frac{2}{\sqrt{1-M^2}}\mathcal{T}(w_a)+\frac{2}{\sqrt{1-M^2}}\sum_{n=1}^{\infty}(\mathcal{T}\mathcal{P}\mathcal{L}\mathcal{H}\mathcal{P})^n\mathcal{T}(w_a).
\end{equation}
Equation~ \eqref{eq:solution to possio equation k=0 seperated} shows that the pressure jump is a linear combination of two terms. The first term is equivalent to the pressure jump along the airfoil in open flow (represented by the first term on the right hand side of equation~\eqref{eq:solution to possio equation k=0 seperated}). Moreover, the second term couples the effect of the downwash with the effect of the elevation from the ground. Next, we require a preliminary lemma in order to establish the existence and uniqueness of the solution to the derived Possio equation \eqref{eq:final form of Possio equation to be solved}.

\begin{lemma}\label{lemma:PLHP bounded}
The  operator $\mathcal{P}\mathcal{L}\mathcal{H}\mathcal{P}$ given by equation \eqref{eq:PLHP} is bounded from $L^{p}[-b,b]$ to $C[-b,b]$ for any $p \geq 1$ and hence into any $L^{p}[-b,b]$ with a bound inversely proportional to $c$.
\end{lemma}
\begin{proof}

Estimating  $\mathcal{P}\mathcal{L}\mathcal{H}\mathcal{P}( f)$ using H\"older's inequality results in the inequality 
\begin{equation}
|\mathcal{P}\mathcal{L}\mathcal{H}\mathcal{P}(f)|
\leq  ||g_{c}||_{L^{q}[-2b,2b]}  \Vert f \Vert_{L^{p}[-2b,2b]} 
\end{equation}
where $\frac{1}{p}+\frac{1}{q}=1$ and $g_c (t)=\frac{t}{t^2+c^2}$ with a norm  estimated by
\begin{equation}\label{eq:bound of PLHP}
 ||g_{c}||_{L^{q}[-2b,2b]}=\left(\int_{-2b}^{2b}\left|\frac{t}{t^2+c^2}\right|^{q}dt \right)^{\frac{1}{q}}\leq\left|\frac{c}{c^2+c^2}\right|(4b)^{\frac{1}{q}}=\frac{(4b)^{\frac{1}{q}}}{2c},
\end{equation}
and that completes the proof.
\end{proof}

Finally, the existence-uniqueness theorem of equation~ \eqref{Final Possio} is introduced and proved.

\begin{theorem}
Given $w_{a}  \in L^{p}[-b,b]$ with $p>2$, there exists a value $c_{0}>0$ such that for all $c\geq c_{0}$ the Possio equation \eqref{Final Possio} has a unique solution $A \in L^{\frac{4}{3}-}[-b,b]$  satisfying the Kutta condition.
\end{theorem}
\begin{proof}
Consider the equivalent equation \eqref{eq:final form of Possio equation to be solved}.
Given $w_{a} \in L^{p}[-b,b]$, we have $\mathcal{T}w_{a} \in L^{r}[-b,b]$ for every $r<4/3$.  
It is sufficient to show that the operator
$\mathcal{T}\mathcal{P}\mathcal{L}\mathcal{H}\mathcal{P}$ is a contraction on the space $L^{r}[-b,b]$.
Now, 
$\mathcal{T}\mathcal{P}\mathcal{L}\mathcal{H}\mathcal{P}
$ is bounded on $L^{r}[-b,b]$ since it is a composition of $\mathcal{P}\mathcal{L}\mathcal{H}\mathcal{P}:L^{\frac{4}{3}-}[-b,b]\rightarrow L^{2+}[-b,b]$ and $\mathcal{T}:L^{2+}[-b,b]\rightarrow L^{\frac{4}{3}-}[-b,b]$ as established in Lemmas \ref{Tricomi bounded} and \ref{lemma:PLHP bounded} respectively. The norm $||\mathcal{T}\mathcal{P}\mathcal{L}\mathcal{H}\mathcal{P}||_{\mathcal{L}(L^{r}[-b,b])}$ is estimated  
using the bounds on $||\mathcal{T}||_{\mathcal{L}(L^{p},L^{r})}$ and 
$||\mathcal{P}\mathcal{L}\mathcal{H}\mathcal{P}||_{\mathcal{L}(L^{r},L^{p})}$, from Lemmas \ref{Tricomi bounded} and \ref{lemma:PLHP bounded} respectively, to obtain
\begin{align*}\label{eq:inequality for the bound of TPLHP}
||\mathcal{T}\mathcal{P}\mathcal{L}\mathcal{H}\mathcal{P}||_{\mathcal{L}(L^{r}[-b,b])} &\leq ||\mathcal{T}||_{\mathcal{L}(L^{p},L^{r})} ||\mathcal{P}\mathcal{L}\mathcal{H}\mathcal{P}||_{\mathcal{L}(L^{r},L^{p})}\\
 & \leq \frac{K 4b^{q}}{2c},
\end{align*}
 where $\frac{1}{p}+\frac{1}{q}= 1$, and $K=||\mathcal{T}||_{\mathcal{L}(L^{p},L^{r})}$.
Choosing $c >c_{0} = \frac{K 4b^{q}}{2}$, we have that 
$||\mathcal{T}\mathcal{P}\mathcal{L}\mathcal{H}\mathcal{P}||_{\mathcal{L}(L^{r}[-b,b])}$ is a contraction, and hence 
\eqref{eq:final form of Possio equation to be solved} has a unique solution $A \in L^{r}[-b,b]$ for any $r<4/3$ or $A \in L^{4/3-}[-b,b]$. Moreover, $A$ clearly satisfies the Kutta condition. 
From Lemma \ref{T}, $A$ must satisfy \eqref{Final Possio}, and on the other hand, a solution to \eqref{Final Possio} which satisfies the Kutta condition must necessarily satisfy \eqref{eq:final form of Possio equation to be solved}. This establishes that 
$A$ is the unique solution to  \eqref{Final Possio} in the class of $L^{\frac{4}{3}-}[-b,b]$ functions satisfying the Kutta condition.
\end{proof}
\subsection{Approximate Solution}
Next, an approximate solution to \eqref{eq:final form of Possio equation to be solved} is obtained. From equation \eqref{eq:PLHP} and \eqref{eq:tricomi operator} and changing integration order, the operator  $\mathcal{T}\mathcal{P}\mathcal{H}\mathcal{P}(A)$ can be written as:
\begin{equation}\label{eq:TPLHP after changing integration order}
\mathcal{T}\mathcal{P}\mathcal{L}\mathcal{H}\mathcal{P}(A)(x)=\frac{-1}{\pi c}\int_{-b}^{b}A(\tau)\left(\sqrt{\frac{b-x}{b+x}}\int_{-b}^{b}\sqrt{\frac{b+t}{b-t}}\frac{1}{t-x} g\left(\frac{\tau-t}{c}\right)dt\right)d\tau.
\end{equation}
The expression in brackets in \eqref{eq:TPLHP after changing integration order} corresponds to applying the Tricomi operator on $g\left(\frac{\tau-t}{c}\right)$. The function $g\left(\frac{\tau-t}{c}\right)$ can be linearly approximated by $\frac{\tau-t}{c}$ and that reduces equation~\eqref{eq:TPLHP after changing integration order} to
\begin{equation}\label{eq:TPLHP after linear approxiamtion}
  \mathcal{T}\mathcal{P}\mathcal{L}\mathcal{H}\mathcal{P}(A)(x)=\frac{-1}{\pi c^2}\int_{-b}^{b}A(\tau)\left(\tau\mathcal{T}(1)(x)-\mathcal{T}(x)(x) \right)\, d\tau,
\end{equation}
but $\mathcal{T}(1)(x)=\sqrt{\frac{b-x}{b+x}}$ and $\mathcal{T}(x)(x)=\sqrt{b^2-x^2}$. Therefore, equation~ \eqref{eq:final form of Possio equation to be solved} is simplified to be:
\begin{equation}\label{eq:approximate equation in terms R0 and R1}
  \frac{2}{\sqrt{1-M^2}}\mathcal{T}(w_{a})=A-\frac{1}{\pi c^2}\left(\sqrt{b^2-x^2}\mathcal{R}_{0}(A)-\sqrt{\frac{b-x}{b+x}}\mathcal{R}_{1}(A)\right),
\end{equation}
where $\mathcal{R}_{0}(A)=\int_{-b}^{b}A(\tau)d\tau$ and $\mathcal{R}_{1}(A)=\int_{-b}^{b}\tau A(\tau) \,d\tau$. Applying $\mathcal{R}_{0}$ and $\mathcal{R}_{1}$ on equation~\eqref{eq:approximate equation in terms R0 and R1} respectively results in the set of equations
\begin{eqnarray}\label{eq:equations for R0}
  \frac{2}{\sqrt{1-M^2}}\mathcal{R}_0((\mathcal{T}(w_a)))=\mathcal{R}_0(A)-\frac{1}{\pi c^2}\left(\mathcal{R}_0\left(\sqrt{b^2-x^2}\right)\mathcal{R}_0(A)-\mathcal{R}_0\left(\sqrt{\frac{b-x}{b+x}} \right )\mathcal{R}_1(A) \right ), \\
 \label{eq:equations for R1}
 \frac{2}{\sqrt{1-M^2}}\mathcal{R}_1((\mathcal{T}(w_a)))=\mathcal{R}_1(A)-\frac{1}{\pi c^2}\left(\mathcal{R}_1\left(\sqrt{b^2-x^2}\right)\mathcal{R}_0(A)-\mathcal{R}_1\left(\sqrt{\frac{b-x}{b+x}} \right )\mathcal{R}_1(A) \right),
\end{eqnarray}
where $\mathcal{R}_0\left(\sqrt{\frac{b-x}{b+x}}\right)=b\pi$, $\mathcal{R}_0(\sqrt{b^2-x^2})=\frac{\pi b^2}{2}$, $\mathcal{R}_1\left(\sqrt{\frac{b-x}{b+x}}\right)=-\frac{\pi b^2}{2}$, and $\mathcal{R}_1(\sqrt{b^2-x^2})=0$. 
Solving equations \eqref{eq:equations for R0} and \eqref{eq:equations for R1} for $\mathcal{R}_0(A)$ and $\mathcal{R}_1(A)$ results in
\begin{align}\label{eq:R0 and R1 expressions}
\mathcal{R}_0(A)&=\frac{2}{\sqrt{1-M^2}\left(1-\frac{b^2}{2c^2} \right )^2}\left(\left(1-\frac{b^2}{2c^2} \right )\mathcal{R}_0\left(\mathcal{T}(w_a)\right)-\frac{b}{c^2}\mathcal{R}_1\left(\mathcal{T}(w_a)\right) \right ), \\
 \mathcal{R}_1(A)&=\frac{2\mathcal{R}_1\left(\mathcal{T}(w_a)\right)}{\sqrt{1-M^2}\left(1-\frac{b^2}{2c^2} \right )}.\nonumber
\end{align}
Finally, substituting \eqref{eq:R0 and R1 expressions} in \eqref{eq:approximate equation in terms R0 and R1} and solving for $A$ yields
\begin{equation}\label{eq:A}
  A=\frac{2}{\sqrt{1-M^2}}\left[\mathcal{T}(w_a)+\frac{1}{\pi c^2\left(1-\frac{b^2}{2c^2}\right)^2} \begin{pmatrix}
\sqrt{b^2-x^2}\\
-\sqrt{\frac{b-x}{b+x}}
\end{pmatrix}^T\begin{pmatrix}
1-\frac{b^2}{2c^2} &-\frac{b}{c^2} \\
 0&1-\frac{b^2}{2c^2}
\end{pmatrix}\begin{pmatrix}
\mathcal{R}_0\left(\mathcal{T}(w_a) \right )\\
\mathcal{R}_1\left(\mathcal{T}(w_a) \right )
\end{pmatrix} \right ].
\end{equation}

\subsubsection{Lift and Moment Calculations}
For the steady state case, the airfoil downwash $w_a$ is a function of the angle of attack $\theta$. In particular, 
\begin{equation}\label{eq:downwash steady state case}
  w_a=-U\theta.
\end{equation}
Note that the downwash given in \eqref{eq:downwash steady state case} is not a function of $x$ and therefore, only $\mathcal{R}_0(\mathcal{T}(1))$ and $\mathcal{R}_1(\mathcal{T}(1))$ are required to compute the pressure jump for the steady state case with a downwash given by \eqref{eq:downwash steady state case}.
The pressure jump is calculated to be
\begin{equation}\label{eq:A seady state}
  A=-\frac{2U\theta}{\sqrt{1-M^2}}\left[\left(1+\frac{b^2}{2\pi c^2\left(1-\frac{b^2}{2c^2}\right)}\right)\sqrt{\frac{b-x}{b+x}}+\frac{b}{\pi c\left(1-\frac{b^2}{2c^2} \right )}\sqrt{b^2-x^2} \right ].
\end{equation}
Now, we calculate the aerodynamic lift $\textbf{F}$ on the airfoil using the formula
\begin{equation}\label{}
 \textbf{F}=-\rho U\int_{-b}^{b}A(x)dx=-\rho U\mathcal{R}_0(A).
\end{equation}
Using equation \eqref{eq:A seady state}, we obtain the following expression for $\textbf{F}$ 
\begin{equation}\label{eq:Lift}
 \textbf{F}=\frac{2\pi\rho U^2b\theta}{\sqrt{1-M^2}}+\frac{\rho U^2b^5\theta}{2c^4\sqrt{1-M^2}\left(1-\frac{b^2}{2c^2} \right)^2}.
\end{equation}
\begin{figure}[!ht]
\includegraphics[width=5in ,keepaspectratio=true]{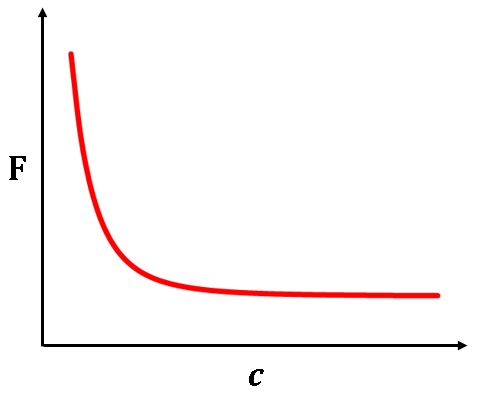}
\centering
\caption{General shape of the lift force profile as a function of the parameter c}
\label{fig:wing structure}
\end{figure}
Note that in equation (\ref{eq:Lift}), the first term on the right hand side represents the aerodynamic lift assuming an open flow and the second term introduces the aerodynamic lift due to the ground effect. Next, the aerodynamic moment $\textbf{M}$ on the airfoil is calculated using the following equation
\begin{equation}\label{}
\mathbf{M}=\rho U\int_{-b}^{b}(x-a)A(x) \, dx=\rho U\mathcal{R}_1(A)+a\mathbf{F},
\end{equation}
where $a$ is the location of the center of rotation of the airfoil.  The aerodynamic moment is given by
\begin{equation}\label{eq:moment}
  \mathbf{M}=\left(\frac{\pi \rho U^2 b^2 \theta}{\sqrt{1-M^2}}+\frac{2a\pi\rho U^2 b \theta}{\sqrt{1-M^2}} \right )+\left(\frac{\rho U^2 \theta b^2}{c^2\sqrt{1-M^2}\left(1-\frac{b^2}{2c^2}\right )}+\frac{a \rho U^2 b^5 \theta}{2c^4\sqrt{1-M^2}\left(1-\frac{b^2}{2c^2} \right )^2} \right ).
\end{equation}

Similar to equation \eqref{eq:Lift}, equation \eqref{eq:moment} separated the open flow moment (first term on the right hand side) from the moment due to the ground effect.
\begin{figure}
\includegraphics[width=5in ,keepaspectratio=true]{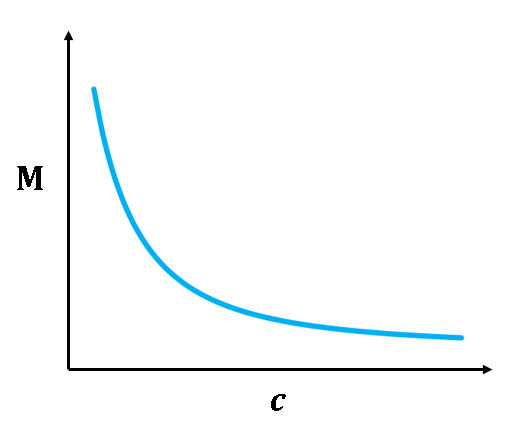}
\centering
\caption{General shape of the moment profile as a function of the parameter c}
\label{fig:wing structure}
\end{figure}
\begin{remark}[Complex analysis approach]
For the steady state incompressible case $(k= 0,M=0)$, the present problem can be studied using complex analytic methods which include Joukowski transformations  to construct  flow potential  and contour integration to obtain the aerodynamic loads \cite{complex analysis}. Despite the relative ease of obtaining flow potential using Joukowski transformations, up to the knowledge of the authors,  there are no closed expressions for the aerodynamic loads in the present case.
\end{remark}
\subsection{Divergence speed}
In this section, the divergence speed $U_{div}$ of a wing structure with thin airfoil (see figure \eqref{fig:wing structure} in subsonic flow is calculated \cite{Transonic Dip, Dowell}.
\begin{figure}[!ht]
\includegraphics[width=5in ,keepaspectratio=true]{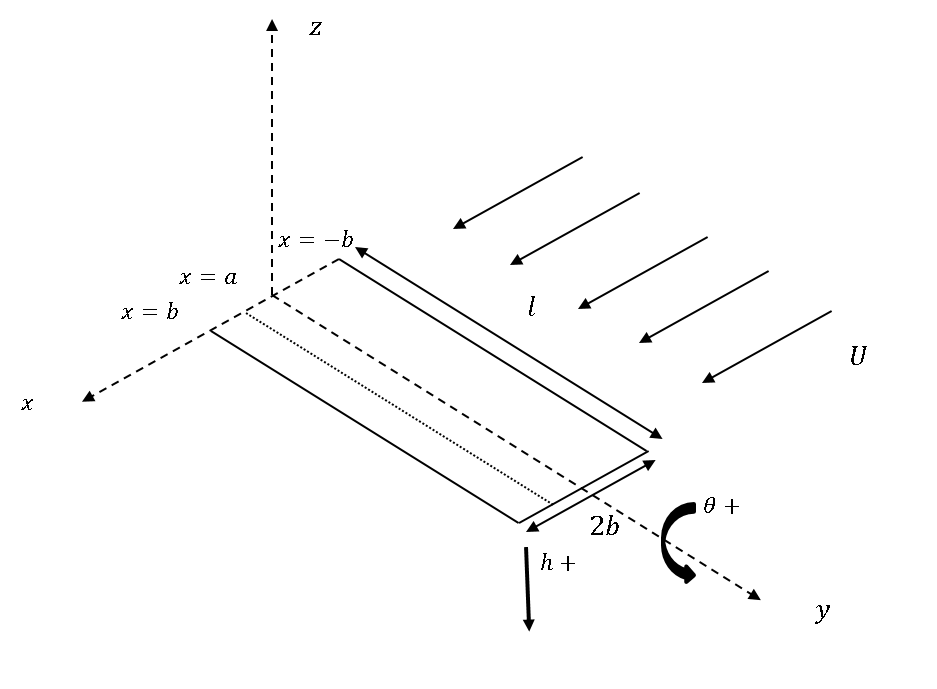}
\centering
\caption{wing configuration}
\label{fig:wing structure}
\end{figure}
The steady state configuration of the wing is governed by the equations
\begin{align}\label{eq:wing equation theta}
-GJ\frac{d^2}{dy^2}\theta&=\mathbf{M},\\
\label{eq:wing equation h}
EI\frac{d^2}{dy^4}h&=-\mathbf{F},
\end{align}
where $0\leq y\leq L$, $L$ is the span of the wing, $\theta(y)$ is the torsion angle and corresponds to the angle of attack in \eqref{eq:downwash steady state case},  $h(y)$ is the deflection of the wing, $GJ$ is the torsional stiffness of the wing,
and $EI$ is the bending stiffness. 

The clamped-free boundary conditions satisfied by $h$ and $\theta$ are given by 
\begin{equation}\label{eq:boundary conditions of the structure}
\theta(0)=h(0)=\left.\frac{d}{dy}\theta\right|_{y=L}=\left.\frac{d^2}{dy^2}h\right|_{y=L}=\left.\frac{d^3}{dy^3}h\right|_{y=L}=0.
\end{equation}
The divergence speed is obtained by solving the eigenvalue problem of finding the minimum free stream velocity that satisfies system \eqref{eq:wing equation theta}--\eqref{eq:wing equation h} and the boundary conditions \eqref{eq:boundary conditions of the structure}. 
The divergence speed is now calculated as the following.  The Moment expression \eqref{eq:moment} is rewritten as 
\begin{equation}\label{}
\mathbf{M}=U^2 \delta \theta,
\end{equation}
where
\begin{equation}\label{}
\delta=\left(\frac{\pi \rho  b^2 }{\sqrt{1-M^2}}+\frac{2a\pi\rho b }{\sqrt{1-M^2}} \right )+\left(\frac{\rho b^2}{c^2\sqrt{1-M^2}\left(1-\frac{b^2}{2c^2}\right )}+\frac{a \rho b^5}{2c^4\sqrt{1-M^2}\left(1-\frac{b^2}{2c^2} \right )^2} \right).
\end{equation}
Then, the governing equation of the torsion is given by
\begin{equation}\label{eq:torsion equation to be solved}
\frac{d^2}{dy^2}\theta+\frac{U^2\delta}{GJ} \theta =0.
\end{equation}
The general solution to equation \eqref{eq:torsion equation to be solved} is 
\begin{equation}\label{}
  \theta(y)=a_1 \sin\left(U \sqrt{\frac{\delta}{GJ}}y\right)+a_2 \cos\left(U \sqrt{\frac{\delta}{GJ}}y\right),
\end{equation}
where $a_1$ and $a_2$ are positive constants that depend on the boundary conditions. For the boundary conditions \eqref{eq:boundary conditions of the structure}, the following relation is obtained
\begin{equation}\label{eq: the formula to calculate the divergence speed}
  \cos\left(U \sqrt{\frac{\delta}{GJ}}L\right)=0\Rightarrow U \sqrt{\frac{\delta}{GJ}}L= \frac{2n+1}{2}\pi, n=\pm1,\pm2,...
\end{equation}
The lowest speed that satisfies equation \eqref{eq: the formula to calculate the divergence speed} (divergence speed) is then given by
\begin{equation}\label{}
U_{div}=\frac{\pi}{2L}\sqrt{\frac{GJ}{\delta}}.
\end{equation}
\begin{remark}[Flutter analysis]
The present work can be extended to obtain lift and moment expressions  for the incompressible transient case $(k\ne 0, M=0)$ and consequently a flutter analysis can be conducted in a fashion similar to \cite{iLiff paper}.  
\end{remark}
\section{Conclusion}
In this paper, the problem of subsonic flow of thin airfoil near the ground is studied. A singular integral equation, namely Possio equation, is derived to relate the pressure jump over the airfoil to its downwash. The existence and uniqueness of the new Possio equation is verified for the case of steady subsonic flow. Moreover, an approximate solution of the derived Possio equation in this particular case is established and closed form expressions of aerodynamic loads and the divergence speed are obtained.

\end{document}